\documentclass[conference]{IEEEtran}
\usepackage[linesnumbered,ruled,vlined]{algorithm2e}
\usepackage{amssymb,amsmath}
\usepackage{caption}
\usepackage{subcaption}
\usepackage{comment}
\usepackage{color,soul}
\usepackage{ragged2e} 
\usepackage{algorithmicx}
\usepackage{algpseudocode}
\usepackage{amsthm}
\usepackage{cite}
\usepackage{array}

\usepackage{cite}

\ifCLASSINFOpdf
   \usepackage[pdftex]{graphicx}
   \graphicspath{{../}{../}}
   \DeclareGraphicsExtensions{.pdf,.jpeg,.png,.jpg}
\else
  \usepackage[dvips]{graphicx}
   \graphicspath{{../}}
   \DeclareGraphicsExtensions{.eps}
\fi

\usepackage[caption=false,font=footnotesize]{subfig}
\begin{document}

\title {Microgrid Revenue Maximization by Charging Scheduling of EVs in Multiple Parking Stations}

\author{\IEEEauthorblockN{Bahram Alinia}
		\IEEEauthorblockA{Department RS2M,
			Institut Mines-Telecom\\
			Telecom SudParis, Evry, France\\
			bahram.alinia@telecom-sudparis.eu}
		\and
		\IEEEauthorblockN{Mohammad H. Hajiesmaili}
		\IEEEauthorblockA{The Chinese University of\\
			Hong Kong, China\\
			mohammad@ie.cuhk.edu.hk}
		\and
		\IEEEauthorblockN{Noel Crespi}
		\IEEEauthorblockA{Department RS2M,
			Institut Mines-Telecom\\Telecom SudParis,
			Evry, France\\
			noel.crespi@mines-telecom.fr}}

\newtheorem{theorem}{Theorem}[section]
\newtheorem{lemma}[theorem]{Lemma}
\newtheorem{proposition}[theorem]{Proposition}
\newtheorem{corollary}[theorem]{Corollary}


\maketitle

\begin{abstract}
Nowadays, there has been a rapid growth in global usage of the electronic vehicles (EV). Despite apparent environmental and economic advantages of EVs, their high demand charging jobs pose an immense challenge to the existing electricity grid infrastructure. In microgrids, as the small-scale version of traditional power grid, however, the EV charging scheduling is more challenging. This is because, the microgrid owner, as a large electricity customer, is interested in shaving its global peak demand, i.e., the aggregated demand over multiple parking stations, to reduce total electricity cost. While the EV charging scheduling problem in single station scenario has been studied extensively in the previous research, the microgrid-level problem with multiple stations subject to a global peak constraint is not tackled. This paper aims to propose a near-optimal EV charging scheduling mechanism in a microgrid governed by a single utility provider with multiple charging stations. The goal is to maximize the total revenue while respecting both local and global peak constraints. The underlying problem, however, is a NP-hard mixed integer linear problem which is difficult to tackle and calls for approximation algorithm design. We design a primal-dual scheduling algorithm which runs in polynomial time and achieves bounded approximation ratio. Moreover, the proposed global scheduling algorithm applies a valley-filling strategy to further reduce the global peak. Simulation results show that the performance of the proposed algorithm is $98\%$ of the optimum, which is much better than the theoretical bound obtained by our approximation analysis. Our algorithm reduces the peak demand obtained by the existing alternative algorithm by $16\%$ and simultaneously achieves better resource utilization.
\end{abstract}


\IEEEpeerreviewmaketitle

\section{Introduction}
\label{sec:intro}
As a result of global warming and environmental concerns through dependence on fossil fuels, there has been a rapid proliferation in deploying renewable energy sources. To promote quick adoption of green renewable energy sources, electrification of vehicles is a trend that has been globally advocated in the recent years. With the significant advantage of Electric Vehicles (EVs) in being an environment friendly product, the global interest for using EVs is rapidly growing, such that global sale of EVs increased by about $80\%$ in $2015$~\cite{EVSales}.

\subsection{Microgrid: Definition, Potential, and Challenges}
\label{sec:introA}
Another common trend in the smart grid era is to utilize the potentials of microgrids and distributedly install renewable energy sources. More specifically, medium or large commercial and industrial energy customers such as universities, headquarters, etc.,  can take control of their own energy consumptions by building a microgrid. In microgrid, different (and mainly renewable) energy sources such as solar panels and wind turbines can partially fulfill the energy demand of a local customer \cite{Minghua2013Sigmetrics}. In addition, microgrids can usually work in ``grid-connected'' mode such that they can acquire their residual demand (total demand subtracted by the local renewable supply) from the external grid. In addition to clean energy production, microgrids can also enjoy more power sustainability and reliability, and the ability to pro-actively manage the energy costs.  

Among different approaches toward managing the microgrid cost, the peak-demand, i.e., the maximum energy drawn over a billing cycle, is an important factor that can significantly impact the total energy cost of the microgrids~\cite{Zhang2015eEnergy}. This is because the real-world pricing scheme for medium and large customers is usually a hybrid time-of-use and peak-based charging model where the peak demand over the billing cycle can significantly impact the total energy cost, e.g., the peak price is often more than $100$ times higher than the on-peak price \cite{PeakPrice}. Considering the total aggregated demand, the peak charge contribution in each billing cycle could be as large as $20\%$ to $80\%$ of the total costs \cite{Xu2014Reducing}. Consequently, a substantial cost reduction could be achieved if microgrid owner can pro-actively control its total peak-demand. 

\subsection{EV Charging Scheduling in Microgrids}
The charging requirement of the microgrid EVs is a portion of the total microgrid electricity demand that plays a significant role in microgrid peak-demand control because of the following two reasons: (i) EV charging jobs contribute significantly in total electricity demand of the microgrid, e.g., currently transportation consumes around $29\%$ of the total energy in the US, while electricity consumes around $40\%$, hence, with the current rapid electrification of the vehicles~\cite{EVSales}, it is clear that the total electricity demand of EVs is considerable, and (ii) EV charging is a flexible and deferrable demand, which mean that the microgrid owner can \textit{schedule} their charging jobs. \textit{Putting together this two issues, the main aim of this paper is to use the deferrable property of EVs and schedule their charging jobs which are a considerable portion of microgrid's electricity demand, so as to intelligently control the peak-demand of the microgrid.}

To respond to the electricity demand of EVs, charging stations are being used where EVs can recourse to charge their battery. There can be few to many charging stations dispersed in the single-owner microgrid, e.g., each for the parking lot of a building in a headquarter campus. In microgrid, all parking stations are usually governed by a single utility provider. In such a situation, the microgrid owner can determine a \textit{global} peak demand constraint for EVs such that the aggregate charging demand in different charging stations is less than this global peak demand.

The charging scheduling of EVs in microgrid with the goal of respecting the \textit{global aggregated peak constraint}, however, is a unique problem which is different from the single station EV charging scheduling. Most of the existing work in the literature, tackle the problem in the single parking station scenario. We refer to Section~\ref{sec:rel} for in-depth discussion. As will be discussed in Section \ref{sec:alg}, the global optimal solution cannot be necessary obtained by separately solving the single station problems. 
On the other hand, there are only a few studies that provide global optimal solution for charging scheduling of EVs \cite{He, Moradijoz, Malhotra}. Despite elegant results, the underlying problems
(cost minimization problem in \cite{He}, optimal station capacity problem in \cite{Moradijoz} and  user convenience maximization problem in \cite{Malhotra}) are different from the problem studied in this paper (see Section~\ref{sec:rel} for details). 

\subsection{Problem and Contributions}
In this paper, we aim to employ the deferrable feature of high demand charging jobs  of microgrid EVs and tackle microgrid EV charging scheduling problem to control the peak-demand of the microgrid. In particular, we assume that in a microgrid, multiple charging stations are available for a set of EVs to get charged. In a general scenario, we assume that the EVs are heterogeneous in terms of availability, charging demand, and valuation for getting charged. Then, the goal is to \textit{select} and \textit{schedule} a subset of EVs such that: (i) the charging demand of the selected EVs are fulfilled; (ii) global and local peak constraints of the microgrid are respected;  and finally, (iii) the total microgrid revenue obtained by  the valuation of the selected EVs is maximized. The underlying optimization problem, however, is a mixed integer linear problem which is a NP-hard problem, then it is difficult to solve in general. In this paper, we tackle this problem by pursuing the approximation algorithm design approach and making the following contributions: 


\begin{itemize}

\item We identify that the EV charging scheduling problem in a single parking station is similar to the problem of scheduling deadline sensitive jobs in a cloud server~\cite{Jain}. Then, we extend the problem for the general case with multiple station taken into account. The formulated problem is a ``time-expanded'' extension of the well-known knapsack problem and hence is NP-hard. 

\item We design a primal-dual scheduling algorithm which runs in polynomial time. We then analyze the performance of the algorithm by constructing a particular form of the linear relaxed version of the mixed integer problem. Based on weak duality property, we obtain the approximation ratio of $\alpha = {\Big( 1+\sum_j {\frac{C_j}{C_j-K_j}}.\frac{s}{s-1}\Big)}$, where $C_j$ is local peak constraint in station $j$, $K_j$ is the maximum charging rate of the EVs in station $j$ and $s$ is a slackness parameter. We highlight that when $C_j \gg K_j$, then $\alpha \approx m$, where $m$ is the number of parking stations. As compared to the proposed solution in~\cite{Jain}, our algorithm has an elegant additional step to apply a valley filling strategy and enhance actual peak of the system without degrading the total revenue obtained. 

\item  We conduct a set of simulations to evaluate the performance of our proposed approximation algorithms. The results for a set of representative runs reveal that the proposed scheduling algorithm is $96\%$ close to the optimum (i.e., optimal total revenue), on average, which is much better than the obtained theoretical approximation bound. Moreover, the algorithm results in $4\%$ higher resource utilization and significantly reduces the global microgrid peak by $16\%$ as compared to the previous study \cite{Jain}.
\end{itemize}


\subsection{Paper Organization}
The rest of this paper is organized as follows. Section~\ref{sec:rel} reviews the literature. In Section \ref{sec_SystemModel} the system model  and problem formulation for charging scheduling problem is described. Section \ref{sec:alg} provides a near optimal scheduling algorithm for the problem and discuss on the approximation bound. The algorithm is evaluated through simulation in Section \ref{sec:simul} and finally Section \ref{sec:conclusion} concludes the paper.
\label{sec_intro}

\section{Relared Work}
\label{sec:rel}

\subsection{EV Charging Scheduling in Single Station Scenario}
For a single parking station, the problem of optimal scheduling has attracted substantial research studies. Most researches addressed optimizing charging cost \cite{Tan, Shao, Vagropoulos} and user convenience level \cite{Wen}. Also, a few studies consider joint optimization of charging for user and aggregator \cite{Jian, Jin}. In \cite{Tan}, the total energy cost for charging station is minimized through load smoothing. The authors propose a solution without relying on future load information and guarantee fulfillment of all demands for a feasible scheduling problem. \cite{Vagropoulos} considers cost minimization for a local EV aggregator in a market for a centralized real-time EV charging management. The optimal solution obtained by appropriately setting charging power of EVs every few seconds during the working hours. \cite{Shao} addressed minimizing grid generation cost for a large population of EVs. The generation cost consists of the fuel, and startup/shutdown cost of EV. A hierarchical mechanism including both grid and EV constraints by adapting Benders cut to coordinate between different levels is proposed but scheduling of EVs is not part of the work. \cite{Wen} addressed the problem of maximizing user convenience. As solution, an optimal subset of EVs is selected while meeting circuit-level demand limits. However, the problem is solved for only one time slot. 

\subsection{EV Charging Scheduling in Multiple Station Scenario} 
Studies in \cite{He,Malhotra,Akhavan} addressed charging scheduling problem in multiple station setting. In \cite{He}, a global scheduling optimization problem in a system consisting of a central controller and multiple local controllers to minimize the total cost for charging EVs is tackled. However, there is no limit on the maximum peak demand which system can tolerate. Consequently, the peak value can be arbitrary high depending on the total submitted demands. This can result in a big billing cost for the microgrid owner and also pose danger for the grid system when the EV penetration level is too high. Besides, the charger devices installed in the parking stations have limitation on the maximum electricity that they can transfer in a time unit \cite{Tesla}. We solve the issue by setting local and global peak constraints. However, to meet the peak constraints, it may not be feasible to respond to all charging demands and consequently, only a subset of EVs can be charged \cite{Akhavan, Lee}. We will apply a priority based selection process (based on the EVs' valuation) to respond to the demands to make sure that the peak varies between desired and safe values. The selection process can be made based on different criteria such as the value of each charging request \cite{Lee, Xiang} or the priority \cite{Akhavan}. As an alternative approach to control peak, some other studies directly targeted minimizing peak \cite{Zhao, Karfopoulos}. In \cite{Zhao}, an intelligent online algorithm is developed for EV charging to minimize the peak load by minimizing the impact of variability and uncertainty of renewable energies inside the grid. \cite{Karfopoulos} considers valley filling by leveraging V2G in peak hours. Although the peak is minimized in these studies, it cannot guarantee that the minimized peak is still desirable and tolerable. 

To avoid big billing cost in peak hours, \cite{Moradijoz} proposed a solution based on genetic algorithm to find optimal capacity and location of parking lots for serving demands in peak hours with the goal of maximizing total benefit of all stations. Although authors study the problem under multiple stations setting, their solution is not applicable when the parking stations are already set up. The study in \cite{Malhotra} is the most relevant work to our case where both local peak and global peak constraints in multiple station scenario are considered. The objective in \cite{Malhotra} is to maximize user convenience level which is different from the aim of this paper.

\section{System Model}
\label{sec_SystemModel}

We consider a set $\mathcal{P}$ of $m$ charging stations all governed by a single utility provider for EV charging purpose. The time horizon is divided to $T$ equal length time slots ${t=1, 2, \dots ,T}$ (e.g., $T=24$ with time slots of $1$ hour length). There is a total number of $n$ EVs denoted by set $\xi$, all available at time $t=1$ to be charged. Each EV $i$ is associated with a deadline $d_{i}$ and a charging demand $D_i$. We assume that the EVs select their charging station and so the assignments are given to the problem. Moreover, $P(i)$ is charging station of EV $i$ where $P$ defined as $P:\xi\rightarrow  \mathcal{P}$. If EV $i$ is charged before its deadline the gain is $v_{i}$ and zero, otherwise i.e., there is no partial credit for partial charging \cite{Akhavan,Xiang,Shroff2014}. For each EV $i$, the battery charging operation can be scheduled to be done in time interval $[1,d_i]$ where the charging rate in each time slot is bounded to $k_i$, a value dependent to physical properties of the EV's battery. Moreover, $K_j, j=1, 2, \dots , m$ denotes the maximum $k_i$ among all EVs in parking station $j$ i.e., $K_j=\max_{P(i)=j} k_i$. It is assumed that 
for each EV $i$, its demand and deadline represent a feasible charging profile with respect to its maximum charging rate $k_i$ and a slackness parameter $s\geq 1$ which is the minimum ratio between an EV's deadline and its minimum charging time. More specifically, we have ${D_i\leq \frac{k_id_i}{s}}$. The slackness parameter is imposed by the microgrid in order to make the charging scheduling flexible. With $s=1$ there is the minimum level of flexibility which enhances by increasing $s$. 
In each time slot, the total electricity drawn from the microgrid by station $j$ and the sum of  electricity consumed by $m$ parking stations are limited to a specific amount $C_j$ and $C_\textrm{total}$, respectively. As explained in Section \ref{sec:introA}, these local peak and global peak constraints are set based on cost effective consumption policy or due to the fact that charger devices has constraint on the maximum electricity that they can output in a time slot \cite{Wen, Xiang, Tesla}. 

\subsection{Problem formulation}

The goal of this section is to formulate an optimization problem to schedule the charging of the EVs with the objective of maximizing the total value gained from charged EVs in $m$ parking lots while respecting local and global peak constraints. The original charging scheduling problem is a mixed linear integer programming since in our model the total resources that each EV $i$ receives is constrained to be $D_i$ if the EV is selected or $0$, otherwise. The original integer problem is a generalized form of the $0-1$ Knapsack problem which is a well-known NP-hard problem. We will give an intuition to understand the similarity of these problems Section \label{subsec:algorithm} however, skip to prove that there is a polynomial time algorithm to reduce $0-1$ Knapsack problem to the original problem due to space limitation and the straightforwardness of the proof. For the ease of solution design, we formulate a relaxed linear version of the problem as follows: \vspace{4mm}


\begin{align}
Z: \ &\max_{} \quad\sum_i \frac{v_{i}}{D_i}\sum_t y_i(t) \label{Pobjec} \\
&\ \textrm{s.t.}\ \ \ \sum_{t} y_i(t)\leq D_i,\quad\quad \forall i \label{ConstPartial}\\
&\quad\quad\ \ \sum_{i} y_i(t)\leq C_{\textrm{total}} ,\quad\quad\forall t \label{ConstTotal}\\
&\quad\quad \sum_{i:P(i)=j} y_i(t)\leq C_j ,\quad\quad\forall t, j\label{ConstLocal}\\
&\quad\quad\ \ y_i(t) - \frac{k_i}{D_i}\sum_{t^\prime} y_i(t^\prime) \leq 0, \quad\quad\forall i, t\label{ConstSpeed}\\
&\quad\quad\ \ y_i(t) \geq 0, \quad\quad\forall i, t\label{ConstPositive}.
\end{align}

In the formulation, $y_i(t)$ is the amount that EV $i$ is charged in time slot $t$. Constraint (\ref{ConstPartial}) ensures that this amount is less than or equal to the requested demand $D_i$.  When we are designing our algorithm, we do not let partial charging. The global and local peak constraints are represented by constraints (\ref{ConstTotal}) and (\ref{ConstLocal}). In each time slot, the charging rate of EV $i$ should be less than or equal to its maximum charging rate i.e., $y_i(t)\leq k_i, \forall i,t$. This is shown in a strength form by constraint (\ref{ConstSpeed}). This form of the constraint is used to reduce the integrality gap of the relaxed linear problem \cite{Carr}. Finally, constraint (\ref{ConstPositive}) is added since the received resource in each time slot should be a non-negative value.

Problem $Z$ is an extension of formulated problem in \cite{Jain} where the resource allocation problem for deadline sensitive jobs in a cloud computing system for a single cloud center is studied. It turns out that the resource allocation problem in cloud systems and EV charging scheduling problem in a \emph{single station} share similar structure. Indeed, each EV charging profile in our scheduling problem can be seen as a task in cloud system with a deadline, value and a maximum CPU time unit to process the job. In this paper, we design a near-optimal scheduling mechanism to solve problem $Z$ with bounded approximation gap where the performance analysis relies on a dual fitting method. For this purpose, we construct the dual form of the problem  as follows:\vspace{4mm}

\begin{align}
\bar{Z}:\ &\min_{} \quad\sum_{i} D_i\alpha _i + \sum_{i=1}^m \sum_{t=1}^T C_i\beta (t)+\sum_{t=1}^T C_{\textrm{total}}\gamma (t)\label{DdualObject}\\ 
&\ \textrm{s.t.}\ \ \alpha _i+\beta (t) +\gamma _i + \pi (t)- \frac{k_i}{D_i}\sum_{t^\prime\leq d_i} \pi _i(t^\prime) \geq \frac{v_{i}}{D_i},\label{ConstDcover} \nonumber\\
& \hspace{50mm}\quad\quad\quad\forall i,t \leq d_i\\
&\quad\quad\quad \alpha _i, \beta _i, \gamma,  \pi _i(t)\geq 0,\quad\quad\forall i,t
\end{align}

In the dual problem, the dual variables $\alpha ,\gamma ,\beta$ and $\pi_i(t)$ are associated with constraints (\ref{ConstPartial}), (\ref{ConstTotal}), (\ref{ConstLocal}) and (\ref{ConstPositive}) in the primal problem, respectively. The approximation algorithm is explained in the next section. 

\section{Scheduling Algorithm}
\label{sec:alg}

\subsection{SCS Algorithm}
\label{subsec:algorithm}
In primal-dual algorithm, the goal is to design an algorithm in a way that it produces a good solution for primal problem (with primal value $\Gamma$) and a feasible solution for the dual problem (with dual value $\Lambda$). Then, assuming that the primal problem is a maximization problem, to prove that the algorithm is $\alpha-$approximation for $\alpha\leq 1$, the important part is to show that $\Lambda\leq\frac{1}{\alpha}\Gamma$. Then, since based on weak duality theorem we have $\Lambda\geq OPT$, it is concluded that $\Gamma \geq\alpha OPT$ where $OPT$ is the optimal value.
We design our scheduling algorithm referred as \emph{Smart Charging Scheduling (SCS)} algorithm based on the basic algorithm proposed in \cite{Jain} for job scheduling problem in cloud computing. Then, we analyze approximation factor of the proposed algorithm for the multiple station problem.  We stress that the proposed solution in \cite{Jain} is only applicable on a single station scenario and solving scheduling problem separately in each station does not guarantee that the final aggregated peak of parking lots is lower than the global peak constraint. SCS algorithm is listed in Algorithm 1.

\vspace{4mm}

\begin{algorithm}
\caption{Smart Charging Scheduling (SCS)}
\label{SCS}
 \DontPrintSemicolon 
\KwIn{$n$ EVs with $v_i, d_i$ and $k_i$ associated with each EV $i$, and $m$ parking stations}

\KwOut{A feasible scheduling of EVs}
\BlankLine

\textbf{initialize}: $y\leftarrow 0, \alpha\leftarrow 0, \beta\leftarrow 0, \gamma\leftarrow 0, \pi\leftarrow 0$

sort charging requests in non-decreasing order of value/demand ratio: $v_1\slash D_1\geq v_2\slash D_2\geq\dots \geq v_n\slash D_n$

\slash\slash\textit{\texttt{Use sorted list to process demands}}

\For{(i=1...n)}{

\slash\slash\textit{\texttt{if enough resources remain for EV $i$}}

\If{($\sum_{t\leq d_i}\min\{\bar{W}(t,P(i)), k_i\}\geq D_i$)}
{
SmartAllocate($i$)
}
\Else {
\If{$(\beta (d_i)=0)$}{
$\beta$\_cover($i$)}
}
}

\For{(i=1...n)}{ 
\If {EV $i$ is not selected}{
ReConsider($i$);
}
}
\end{algorithm}
\vspace{4mm}
In the algorithm, $W(t,P(i))=\sum_{i': P(i')=P(i)} y_{i'}(t)$ is the total workload at time slot $t$ in the parking station $P(i)$ and $\bar{W}(t,P(i))$ is the total available load to allocate in time slot $t$ in station $P(i)$. We always have ${\bar{W}(t,P(i))+W(t,P(i))=C_{P(i)},  \forall t, i}$. 

The SCS algorithm works in two phases. In the first phase it sorts the charging requests based on their marginal value $\frac{v_i}{D_i}$ in a non-decreasing order. Then, starting from the top of the sorted list, if the remaining resource is enough for fully responding the current EV's demand, the charging operation will be scheduled. Otherwise, there will be no charging and the next demand in the list will be processed. More precisely, when processing the charging request for EV $i$, the algorithm checks for the feasibility of allocating $D_i$ units of electricity resources to EV $i$ before its deadline $d_i$ without violating maximum charging rate constraint $k_i$ in each time slot (Line 6). If feasibility check passed, SCS calls sub-procedure SmartAllocate($i$) to allocate required resources in interval $[1,d_i]$. After allocation, SmartAllocate($i$) sets $\alpha _i$ to $v_i\slash D_i$ in order to cover dual constraint in Equation (\ref{ConstDcover}). 

If there is not enough resources to charge EV $i$ to its demand $D_i$, no charging will be done. However, we still need to satisfy constraint (\ref{ConstDcover}) in dual problem. For this purpose, the $\beta$\_cover($i$) algorithm inside the main procedure is called when EV $i$ is not a candidate of getting charged. Note that to cover the constraint (\ref{ConstDcover}) for EV $i$, $\beta (t)$ for $t=1,\dots ,d_i$ should be greater than or equal to $\frac{v_i}{D_i}$. $\beta$\_cover($i$) sets $\beta (t)$ to $\frac{v_i}{D_i}$ for all time slots $t$ in interval $[t_\textrm{cov}, R(d_i)]$ (Line 3 and 4 of the algorithm). Observe that when $t_\textrm{cov}>1$ we have $\beta (t^\prime)\geq\frac{v_i}{D_i}, \forall t^\prime< t_\textrm{cov}$ considering that the demands are sorted in a non-increasing order according to their marginal value and $\beta (t')$ is already set to $\frac{v_{i'}}{D_{i'}}$ when processing the earlier charging demand of EV $i'$ in the list which is not selected. Hence, $\frac{v_{i'}}{D_{i'}}\geq\frac{v_i}{D_i}$ and we have $\beta (t)\geq\frac{v_i}{D_i},\ \forall t\in [1,d_i]$ and the dual constraint is satisfied.

Lines $1-4$ of algorithm $\beta$\_cover($i$) is enough to cover dual constraint. However, for any already selected EV $i'$, the algorithm continues in Lines $5-8$ by setting a variable $\Phi_{i'}(t)$ for time slots $t=1,\dots ,R(d_i)$ to a value dependent to amount of the resource that EV $i'$ received in time slot $t$. The value of $\Phi_{i'}(t)$ will be used to analyze the approximation factor of the main algorithm in Section \ref{subsec:analysis} and has no effect on the real scheduling of EVs. We borrowed algorithm $\beta$\_cover($i$) from \cite{Jain} and adjusted it for multiple station mode. 

When EV $i$ is selected to be charged, SmartAllocate($i$) is called to allocate resources. In resource allocation phase, SmartAllocate($i$) applies two main policies: 1) \emph{flat allocation} and, 2) \emph{right-to-left allocation}. With flat allocation, time slots with more available resources i.e., $\bar{W}(i,P(i))$ are preferred to be used for charging purpose. This is in fact a \emph{valley-filling} policy which helps to reduce final peak of the SCS algorithm. It also can be seen as a smoothing method which tries to reduce variance of the allocated resources in different time slots. The simulation results in Section \ref{sec:simul} will confirm this claim. Right-to-left allocation is used when two or more time slots are equal based on the remaining resource. When this policy applies on EV $i$, any EV $i'$ with $i'>i$ and $d_{i'}<d_i$ is more likely to be acceptable for charging since the algorithm tends to charge EV $i$ in time slots in interval $[d_{i'}+1,d_i]$ and keep resources in $[1,d_{i'}]$ for EV $i'$. A ranking based approach is used to apply the aforementioned policies. 
For charging  EV $i$, we rank time slots in interval $[1,d_i]$. Then, charging is done by allocating resources from the higher ranked time slot to lowest one. The rank of a time slot $t$ is obtained based on remaining resources in time slot $t$ (flat allocation) and value of $t$ (right-to-left allocation). 

\vspace{4mm}
\begin{algorithm}
\caption{SmartAllocate($i$)}
\label{SmartAllocate}
 \DontPrintSemicolon 
\KwIn{EV $i$ to be charged}

Rank time slots in interval $[1,d_i]$ such that for $t_1$ and $t_2$:
$rank(t_1)>rank(t_2)$ iff $\bar{W}(t_1,P(i))>\bar{W}(t_2,P(i))$ OR $\bar{W}(t_1,P(i))==\bar{W}(t_2,P(i)) \wedge t_1>t_2$ 

\While {EV $i$ is not fully allocated}{

Select time slot $t$ with highest rank which is not selected before

Allocate $\min \{k_i,\bar{W}(t,P(i))\}$ resources for EV $i$ in time slot $t$

}

$\alpha _i\leftarrow v_i\slash D_i$
\end{algorithm}
\vspace{5mm}

\begin{algorithm}
\caption{$\beta$\_cover($i$)}
\label{BetaCover}
 \DontPrintSemicolon 
\KwIn{EV $i$ which is not selected to charge}

\BlankLine

$t_\textrm{cov}\leftarrow \min \{t: \beta (t)=0\}$

$R(d_i)=\max \{t\geq d_i : \forall t^\prime\in (d_i,t], \bar{W}(d_i)<K_{P(i)} \}$

\For{$(t=t_{cov}\dots R(d_i))$}{
	$\beta (t)\leftarrow v_i\slash D_i$
}

\For{$(t=1\dots R(d_i))$}{
	\For{$(i'=1\dots n))$}{	
	\If{$y_{i'}(t)>0 \wedge \Phi_{i'}(t)=0$}{
		$\Phi_{i'}(t)\leftarrow\big[ \frac{C_{P(i)}}{C_{P(i)}- k_i}\frac{s}{s-1}\big] .\frac{v_i}	{D_i}y_{i'}(t)$
	}
	}
}
\end{algorithm}

\begin{algorithm}
\caption{ReConsider($i$)}
\label{Replace}
 \DontPrintSemicolon 
\KwIn{EV $i$}
\KwOut{Updated schedule}

Let $L$ be an empty list 
 
$v_\textrm{inc}\leftarrow v(i)$

Define $\delta_i(t)=\min \{k_i,\bar{W}(t,P(i))\}$ 

$\Delta\leftarrow \sum_{t\leq d_i}\min\{\bar{W}(t,P(i)), k_i\}, t=1,\dots d_i$

\For{($i^\prime =i-1\dots 1$)}{ 
	\If {EV $i'$ is selected $\wedge P(i')=P(i) \wedge v_\textrm{inc}-v_{i'}>0$}{
		Add EV $i'$ to list $L$	
		
		$v_\textrm{inc}\leftarrow v_\textrm{inc}-v_{i'}$
		
		\For{($t=1\dots d_i$)}{
			$\delta_i(t)\leftarrow \min\{k(i),\delta_i(t)+y_{i'}(t)\}$
		}
	}
}

\If {$\sum_{t\leq d_i}\delta_i(t)\geq D_i$}{
	Remove EVs in list $L$ from charging schedule 
	
	SmartAllocate($i$)
}
\end{algorithm}

In the second phase of SCS, the algorithm tries to increase the total value of selected EVs by calling ReConsider($i$) on every unselected EV $i$ (Line $11-13$). We note that, if in the scheduling problem we set $T=1, m=1, k_i=C_{\textrm{total}}$ and $ d_i=1 \ \forall i=1,\dots ,n$, then the problem becomes equal to the well-known 0-1 knapsack problem. Filling the knapsack by using the same greedy approach (sorting objects based on their marginal value) has the issue that the approximation factor can be arbitrary bad. For example, consider a knapsack problem with two objects where $v_1\slash D_1>v_2\slash D_2$ by having $v_1=2, v_2=C_\textrm{total}, D_1=1, D_2=C_\textrm{total}$. To maximize the total value, the optimal solution here is to choose object $2$ while greedy algorithm selects object $1$ which results in a worst-case approximation factor of $\frac{c}{OPT}$ in general where $c$ is constant. To avoid this problem, one approach is \emph{re-considering} unselected objects after running greedy algorithm and replacing some selected objects in the knapsack with unselected ones and then, test if the result is improved or not. In a simple case, only the largest unselected object can be tested which makes a significant theoretical improvement by providing a worst case approximation factor of $\frac{1}{2}OPT$ . SCS algorithm uses the same idea but with a more intelligent replacing method ReConsider($i$). ReConsider($i$) is called on unselected EVs in the sorted list by reconsidering first the more valuable unselected demands. The algorithm tries to charge an unselected EV by removing some already selected EVs where the total valuation of the selected EVs is less than the value of the unselected EV. If this case occurs, replacement is confirmed which improves the total value of selected EVs. ReConsider($i$) searches for candidates to be replaced among the less valuable requests by processing the sorted list from the end. After that SCS calls ReConsider($i$) on all unselected EVs $i$ following the greedy based selection process, an improvement in the total revenue of finally selected EVs is expected. This will be confirmed by the simulation results in Section \ref{sec:simul}.

\subsection{Analysis}
\label{subsec:analysis}

To analyze performance of the SCS algorithm, we first note that the designed scheduling algorithm respects the constraints in the primal problem. Also, the algorithm produces a feasible solution for the dual problem by covering the dual problem constraint in (\ref{ConstDcover}) through setting $\alpha _i$ to $\frac{v_i}{D_i}$ when EV $i$ is accepted and, $\beta (t)$ to a value greater than or equal to $\frac{v_i}{D_i}$ for $t\leq d_i$ (according to the discussion in Section \ref{subsec:algorithm}) if EV $i$ is not selected. Now, if we bound the total covering cost of the dual constraints, we can obtain an approximation factor for the algorithm. 

\begin{theorem}
SCS algorithm is $\Big( 1+ \sum_{j=1}^m {\frac{C_j}{C_j-K_j}}.\frac{s}{s-1}\Big)$-approximation.
\end{theorem}

\begin{proof}
We sum up all costs of covering dual constraints and then provide a bound for it. Each EV is either selected or not selected.

For the unselected EVs, $\sum_{j=1}^m \sum_{t=1}^T C_j\beta (t)$ determines the cost. When $\beta$\_cover($i$) is running as a result of charging request disapproval of EV $i$, for any previously accepted request $i'$ the algorithm sets $\Phi_{i'}(t)$ to a value proportional to $y_{i'}(t)$ for $t\leq R(d_i)$ (Line 8 of $\beta$\_cover($i$) algorithm). The followings are proved in \cite{Jain} for a single station $i$:

\begin{equation}
\sum_{i'\in S} \sum_{t\leq d_{i'}} \Phi_{i'}(t) \leq \big[ \frac{C_i}{C_i-K_i}.\frac{s}{s-1} \big].\sum_{i'\in S}v_{i'}
\label{EqCharge}
\end{equation} 

\begin{equation}
\sum_{t=1}^T C_i\beta (t)\leq \sum_{i'\in S} \sum_{t\leq d_{i'}} \Phi_{i'}(t).
\label{EqBound}
\end{equation}

For $m$ parking stations, we can obtain the following inequality based on (\ref{EqCharge}) and (\ref{EqBound}),

\begin{align}
\sum_{i=1}^m \sum_{t=1}^T C_i\beta (t)&\leq \sum_{i=1}^m \Big( \frac{C_i}{C_i-K_i}.\frac{s}{s-1}\sum_{i'\in S, P(i')=P(i)}v_{i'}\Big)\nonumber\\\label{EqMS}.
\end{align}

Now for notation convenience let's define $a_i, b_i$ and $F$ as follows:

\begin{align} 
a_i=\frac{C_i}{C_i-K_i}.\frac{s}{s-1}\nonumber\\ b_i=\sum_{j\in S, P(j)=P(i)}v_j,\nonumber\\ F=\sum_{i=1}^m b_i=\sum_{i\in S} v_i\nonumber. 
\end{align}

We can write the right hand side of (\ref{EqMS}) as follows:

\begin{align}
\sum_{i=1}^m a_ib_i&=\sum_i \Big[ a_i(F-\sum_{j\neq i}b_j)\Big]\nonumber\\
&=\sum_{i=1}^m a_iF - \sum_{i=1}^m \Big[ a_i\sum_{j\neq i}b_j \Big]\nonumber\\
&=\sum_{i=1}^m a_iF - \sum_{i=1}^m \Big[ a_i\big( F-b_i\big) \Big]\nonumber\\
&\leq\sum_{i=1}^m a_iF - \sum_{i=1}^m \Big[ a_i\big( F-\max_j \{b_j\} \big) \Big]\nonumber\\
&=\max_i\{b_i\} \sum_{i=1}^m a_i\label{Eqabf}
\end{align}

From (\ref{EqMS}) and (\ref{Eqabf}) we have, 

\begin{align}
\sum_{i=1}^m \sum_{t=1}^T C_i\beta (t)&\leq \max_i\{b_i\}\sum_{i=1}^m{\frac{C_i}{C_i - K_i}}.\frac{s}{s-1}
\end{align}

For the selected EVs, the covering cost is determined by the term $\sum_{i} D_i\alpha _i$ in the dual objective which equals to $\sum_{EV i\in S} v_i$ where $S$ is the set of selected EVs. Therefore, the total cost of covering dual constraints equals to

\begin{align}
&\sum_{i:EV i\in S}v_i + \max_i\{b_i\}\sum_i{\frac{C_i}{C_i-K_i}}.\frac{s}{s-1}\nonumber\\
&\leq \sum_{i:EV i\in S}v_i + \sum_{i:EV i\in S}v_i\sum_i{\frac{C_i}{C_i-K_i}}.\frac{s}{s-1}\nonumber\\
&=\Big[ 1+\sum_i{\frac{C_i}{C_i-K_i}}.\frac{s}{s-1}\Big] \sum_{i:EV i\in S}v_i\label{Eq:bound}
\end{align}
\end{proof}

\section{Simulation Results}
\label{sec:simul}

\begin{figure*}[t!]	
	\centering
	\begin{subfigure}[b]{0.3\textwidth}
		\begin{center}
			\includegraphics[width=\textwidth]{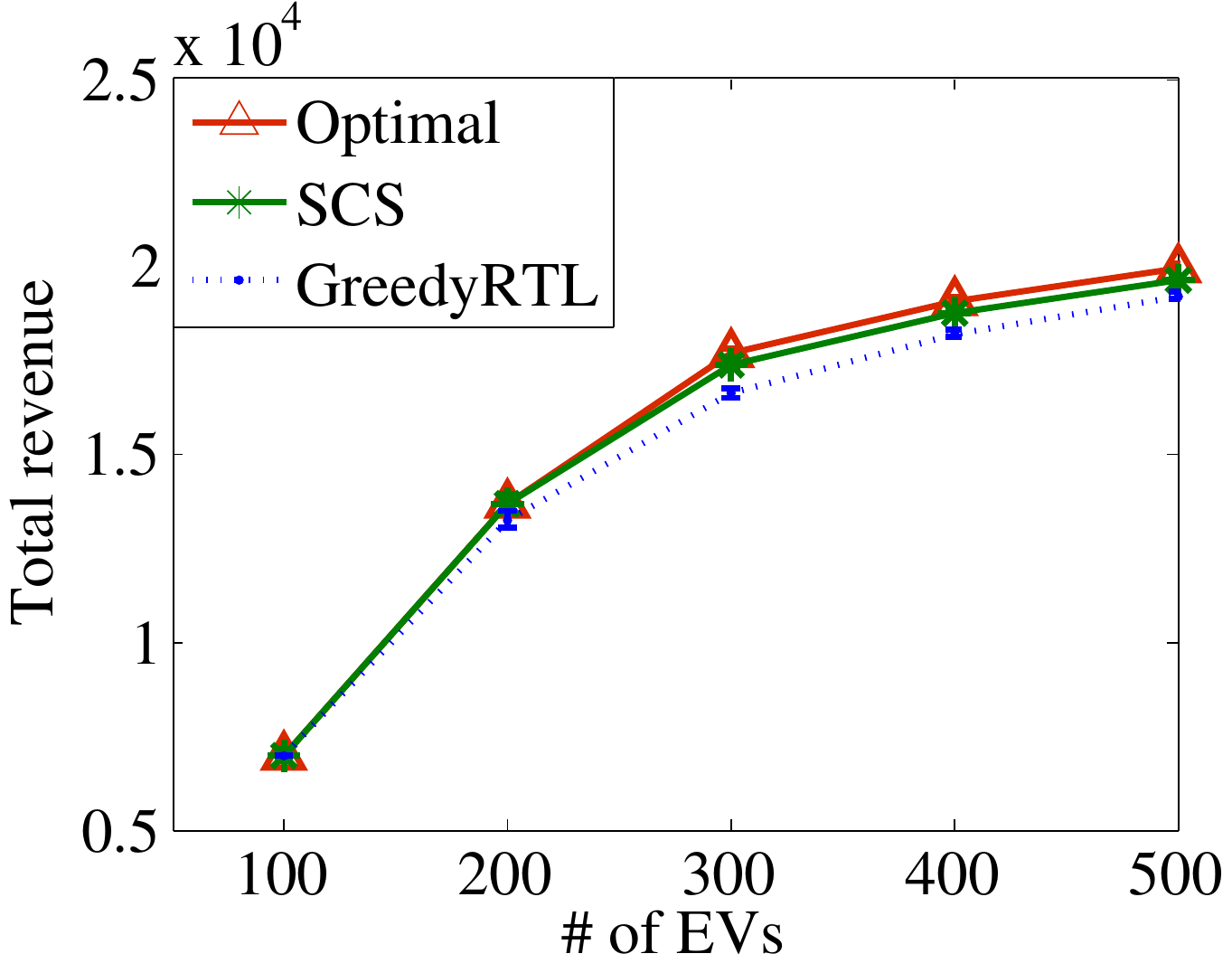}
			\caption{Total revenue}
			\label{fig:v-N}
		\end{center}
	\end{subfigure}%
	~ 
	\begin{subfigure}[b]{0.3\textwidth}
		\begin{center}
			\includegraphics[width=\textwidth]{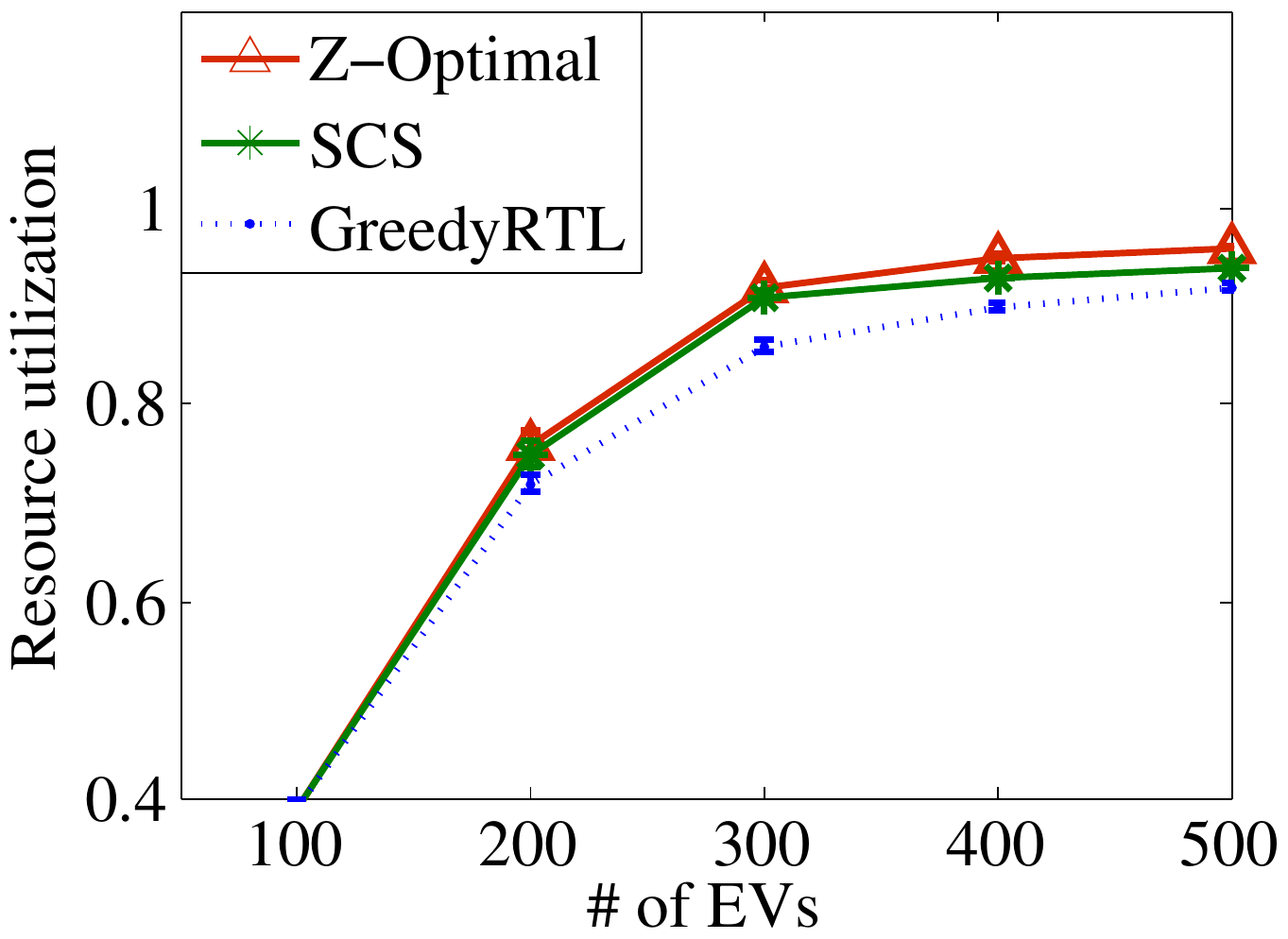}
			\caption{Utilization}
			\label{fig:ut-N}
		\end{center}
	\end{subfigure}%
	\begin{subfigure}[b]{0.3\textwidth}
		\begin{center}
			\includegraphics[width=\textwidth]{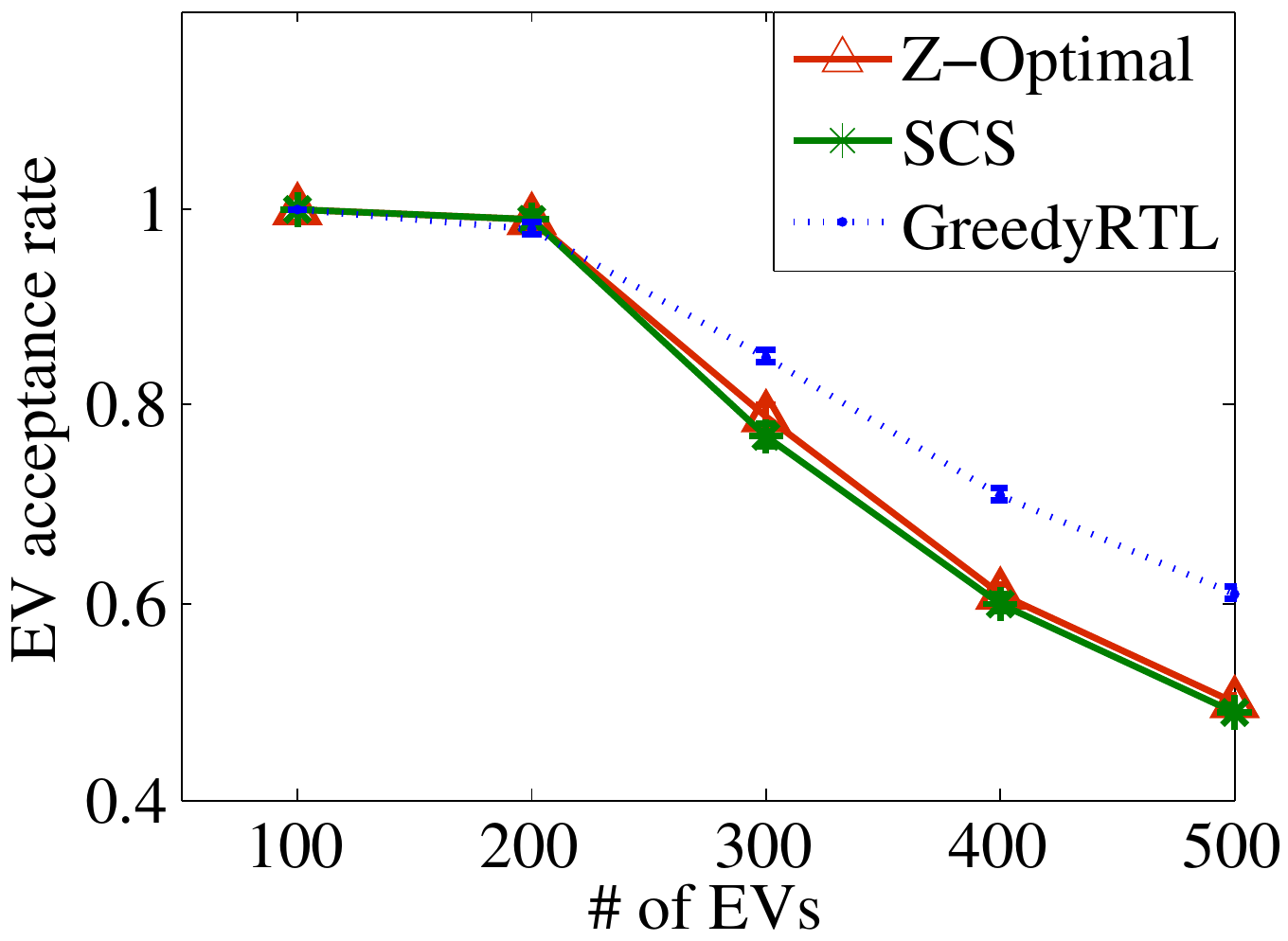}
			\caption{Acceptance rate}
			\label{fig:acc-N}
		\end{center}
	\end{subfigure}%
	\vspace{3mm}

	\caption{Comparison based on EV populations size for revenue, utilization and acceptance rate.} 
	\label{fig:EV_number}
\end{figure*}

In this section, we perform simulation studies to evaluate the performance of our proposed scheduling algorithm. The default simulation setting is as follows. We consider charging scheduling of EVs during a day divided to 24 time slots of length 1 hour. There are 4 parking stations and 200 EVs where EVs are randomly assigned to stations which gives an average of 50 EVs per station. The deadlines are randomly generated based on the assumptions that EV owners usually prefer to take their cars in some special time slots including $07:00$ a.m. to $09:00$ a.m., $12:00$ p.m. to $02:00$ p.m. and $04:00$ p.m. to $07:00$ p.m. The maximum charging rate parameter $k_i$, is set randomly between 1 and 20 ($kWh$), and the feasible demands randomly generated respecting the deadlines and slackness parameter which is set to 1.5. In each time slot, the total peak constraint is set to $500 kWh$ and local peaks are all equal to $125 kWh$. In the simulation figures, the results are plotted with a $95\%$ confidence level and each data point represents average result of 50 random scenarios. In the figures, Z-Optimal, SCS and GreedyRTL refer to optimal solution of the problem $Z$, SCS algorithm proposed in this paper, and method of \cite{Jain}, respectively. Note that since GreedyRTL algorithm works on single station mode, in multiple station setting we run the algorithm separately in each station and let $\sum_{i=1}^m C_i\leq C_{\textrm{total}}$ otherwise, the algorithm may produce infeasible solution by violating total peak constraint. The measured performance metrics include normalized revenue (i.e., $\sum_{i\in S}v_i\slash \sum_{i=1}^n v_i$), resource utilization (i.e., $\sum_{i\in S} D_i/(T*C_\textrm{total})$), demand acceptance rate (i.e., ratio of number of selected EVs to total number of EVs) and actual total peak which should be apparently less than total peak constraint. 

\subsection{Evaluation under different number of EVs}

Fig. (\ref{fig:EV_number}) depicts some results for three algorithms when the total EV numbers increases from 100 to 500 with step size 100. The general trend is that by increasing EV population, total revenue and resource utilization increase according to Fig. \ref{fig:v-N} and Fig. \ref{fig:ut-N}. This is because with increased number of EVs, more electricity is flowed to the EVs which in turn increases the total revenue. However, the rate of accepted charging requests decreases by increasing number of EVs in Fig. \ref{fig:acc-N} as a consequence of constrained peak approach. SCS and GreedyRTL show a close-to-optimal performance based on total value of selected EVs in Fig. \ref{fig:v-N}. In this scenario, SCS is slightly better than GreedyRTL. For the resource utilization in Fig. \ref{fig:ut-N} we have the same trend of Fig. \ref{fig:v-N} and SCS improves GreedyRTL by $4\%$ on average. While SCS improves GreedyRTL in terms of both total revenue and resource utilization, the acceptance rate of GreedyRTL is $6\%$ higher than SCS, on average. To be able to show the performance difference of the three methods, we used raw values of total revenue here. However, in the remaining of this section a normalized value between $0$ and $1$ will be used. 

\begin{figure*}[t!]
	\centering
	\begin{subfigure}[b]{0.3\textwidth}
		\begin{center}
			\includegraphics[width=\textwidth]{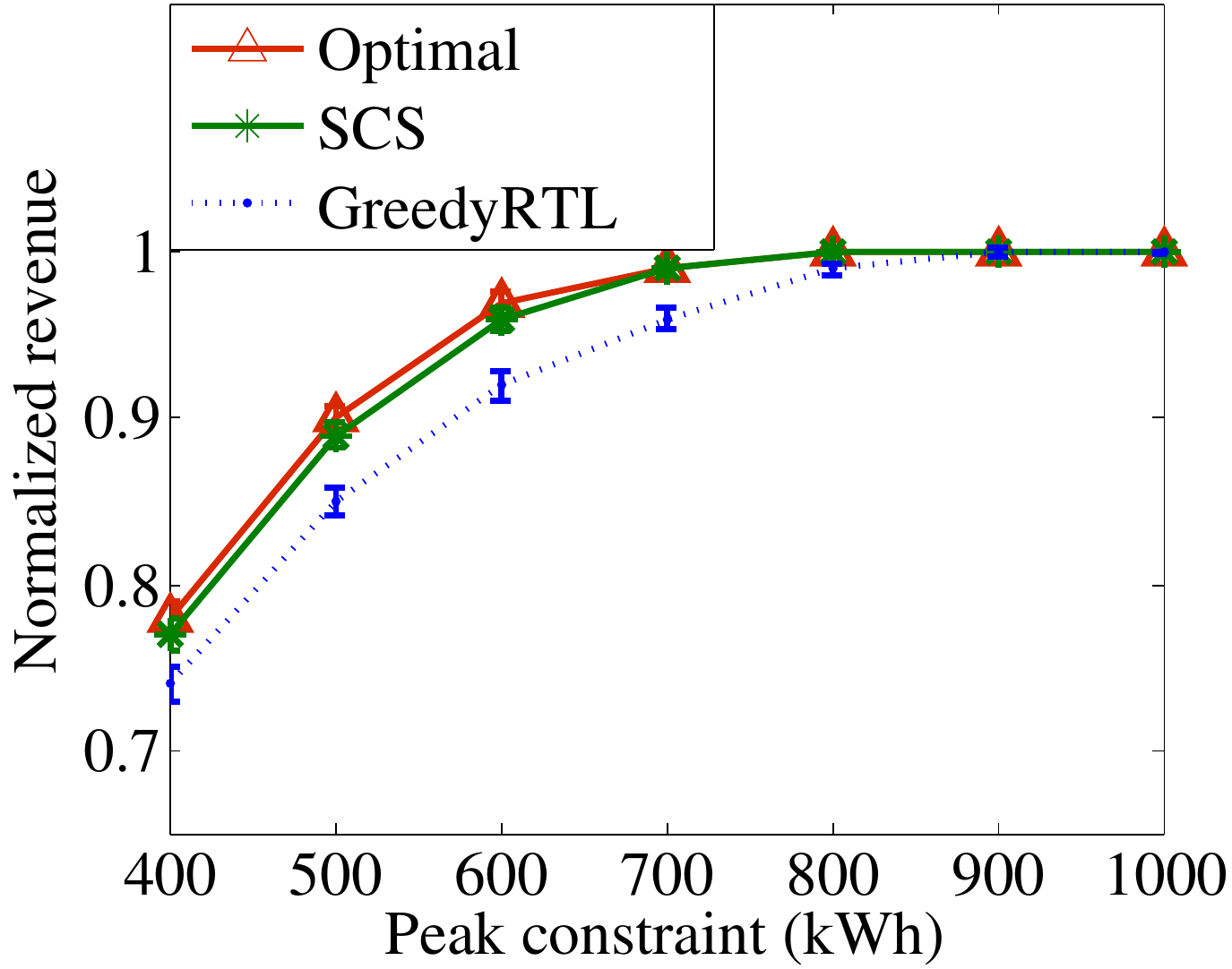}
			\caption{Normalized revenue vs. $C_{\textrm{total}}$}
			\label{fig:v-c}
		\end{center}
	\end{subfigure}%
	~ 
	\begin{subfigure}[b]{0.33\textwidth}
		\begin{center}
			\includegraphics[width=\textwidth]{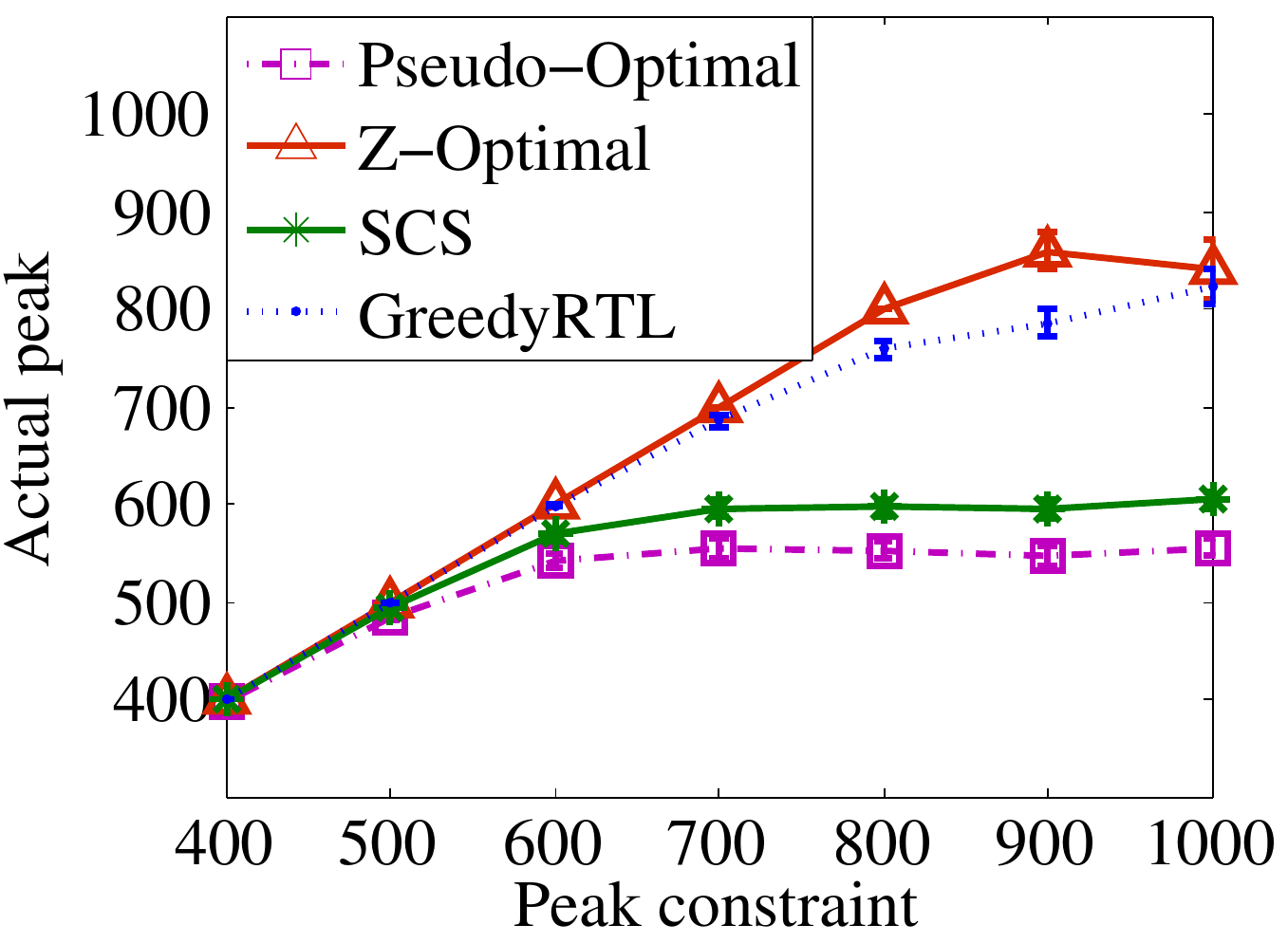}
			\caption{Actual peak vs. $C_{\textrm{total}}$}
			\label{fig:p-c}
		\end{center}
	\end{subfigure}%
	\vspace{4mm}  

	\caption{Improving actual peak by SCS without affecting total gain.} 
	\label{fig:Ex2}
\end{figure*}

\begin{figure}[t!]
	\centering
	\includegraphics[width=.3\textwidth]{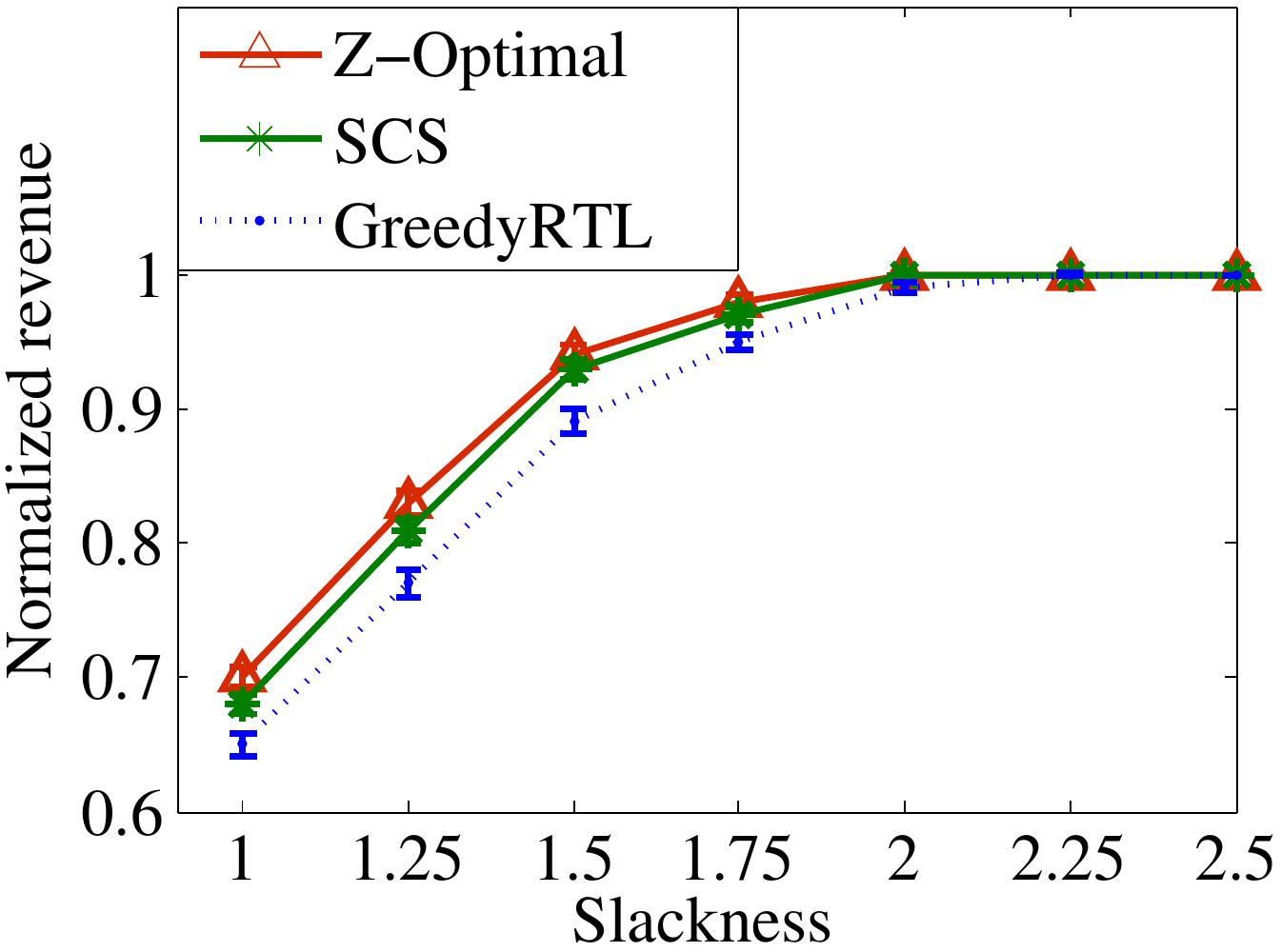}
	\caption{The effect of slackness parameter}
	\label{fig:slackness}	
\end{figure}

\subsection{Comparing actual peaks}

The constraint set in problem $Z$ assure that any feasible solution meets the total peak constrained i.e., in each time slot, the sum of total electricity consumed by stations is less than or equal to $C_\textrm{total}$. An efficient scheduling algorithm may go beyond this by keeping the peak in lower values and not only satisfying the constraint. As it is explained in Section \ref{sec:alg}, SCS applies flat allocation policy to obtain a better peak value. To show how this policy improves peak, we conducted a set of simulations and extracted both normalized revenue (related to main objective of the problem) and actual peak (as secondary performance factor) by varying total peak constraint from $400$ to $1000$ with step size $100$. In Fig. \ref{fig:v-c} it can be observed that SCS not only improves GreedyRTL in terms of problem objective, but also significantly decreases  peak value compared to both GreedyRTL and Z-Optimal with $16\%$ and $18\%$, respectively. In the figure, Pseudo-Optimal refers to the minimum value of the peak in optimal solution space of problem Z. Therefore, no better peak value than Pseudo-Optimal can be reached when the objective value of the problem $Z$ is maximized. In terms of peak value, SCS algorithm is $94\%$ close to Pseudo-Optimal, on average.

\subsection{The effect of slackness parameter}
To give the charging scheduler more flexibility and increase the gains, a slackness parameter $s\geq 1$ is used which is determined by microgrid owner. Increasing the slackness parameter is basically equal to extending EVs deadline since we allow the scheduler to finish the charging with $(s-1)*100$ percentage of delay. Therefore, we expect that increasing slackness parameter results in increased revenue. This is showed in Fig. \ref{fig:slackness} with normalized total revenue.

\bibliographystyle{IEEEtranS}

\section{Conclusion}
\label{sec:conclusion}
\label{sec:conclusion}
In this paper, we consider a electricity grid system with $m$ EV parking stations and propose a centralized scheduling mechanism to optimize charging rates and thus, maximizing total revenue for the grid. The scheduling algorithm is a ${\Big( 1+\sum_{i=1}^m {\frac{C_i}{C_i-K_i}}\frac{s}{s-1}\Big)}$-approximation solution where $C_i$ is peak constraint in station $i$, $K_i$ is the maximum charging rate and $s$ is a slackness parameter. Our solution not only provides theoretical bound on the optimality gap and approximates the optimal objective (i.e., total revenue) with an accuracy of 96\% and extends the previous studies, the average peak is 90\% of the minimum peak in optimal solution space of the main problem.
 
	\bibliographystyle{IEEEtranS}

	\bibliography{IEEEabrv,references}
	
\end{document}